\documentclass[10pt,journal,a4paper]{IEEEtran}
\usepackage{amssymb}
\usepackage{amsmath}
\usepackage{bm}
\usepackage{amsthm}
\usepackage{graphicx}
\usepackage{subfigure}
\usepackage{float}
\usepackage{cite}
\usepackage{enumerate}
\usepackage{enumitem}
\usepackage{multirow}
\usepackage{amsfonts}
\usepackage{color, soul}
\usepackage{url}
\usepackage{amsthm}

\begin{document}
\newtheorem{remark}{\bf~~Remark}
\newtheorem{proposition}{\bf~~Proposition}
\newtheorem{lemma}{\bf~~Lemma}
\newtheorem{theorem}{\bf~~Theorem}
\newtheorem{definition}{\bf~~Definition}
\title{Reconfigurable Refractive Surfaces: An Energy-Efficient Way to Holographic MIMO}

\author{
\IEEEauthorblockN{
	{Shuhao Zeng}, \IEEEmembership{Student Member, IEEE},
	{Hongliang Zhang}, \IEEEmembership{Member, IEEE},
	{Boya Di}, \IEEEmembership{Member, IEEE},\\
	{Haichao Qin}, {Xin Su}, {and Lingyang Song}, \IEEEmembership{Fellow, IEEE}}
	\vspace{-0.9cm}
	\thanks{S. Zeng, B. Di, and L. Song are with Department of Electronics, Peking University, Beijing, China (email: \{shuhao.zeng, boya.di, lingyang.song\}@pku.edu.cn).}
	\thanks{H. Zhang is with Department of Electrical and Computer Engineering, Princeton University, USA (email: hz16@princeton.edu).}
	\thanks{H. Qin and X. Su are with CICT Mobile Communication Technology Co., Ltd. (email: \{qinhaichao,suxin\}@cictmobile.cn).}
}


\maketitle

\begin{abstract}

Holographic Multiple Input Multiple Output~(HMIMO), which integrates massive antenna elements into a compact space to achieve a spatially continuous aperture, plays an important role in future wireless networks. With numerous antenna elements, it is hard to implement the HMIMO via phased arrays due to unacceptable power consumption. To address this issue, reconfigurable refractive surfaces~(RRS) is an energy efficient enabler of HMIMO since the surface is free of expensive phase shifters. Unlike traditional metasurfaces working as passive relays, the RRS is used as transmit antennas, where the far-field approximation does not hold anymore, urging a new performance analysis framework. In this letter, we first derive the data rate of an RRS-based single-user downlink system, and then compare its power consumption with the phased array. Simulation results verify our analysis and show that the RRS is an energy-efficient way to HMIMO. 
\end{abstract}
\vspace{-0.2cm}
\vspace{-0.6cm}
\section{Introduction}
\vspace{-0.1cm}
 One of the key enablers of future wireless networks is Holographic Multiple Input Multiple Output~(HMIMO), where a large number of tiny antennas or reconfigurable elements are integrated into a compact space to achieve a spatially continuous aperture~\cite{ATL_2020}. {Compared with conventional massive MIMO, the HMIMO has a larger dimension with smaller element spacing, and thus can achieve higher spectral efficiency and spatial resolution.} However, it is hard for traditional phased arrays to realize the HMIMO, since hundreds of high-resolution phase shifters are required, leading to unacceptable power consumption~\cite{L_2020}. 
 
 Recently, metasurface-based antennas, which are also referred to as reconfigurable refractive surfaces~(RRS)~\cite{SHBYZHL_2021}, provide a promising solution for the implementation of HMIMO due to their high energy efficiency. An RRS is an ultra thin surface inlaid with a large number of sub-wavelength elements, each of which can refract incident electromagnetic~(EM) waves and apply a tunable phase shift. By controlling the biased voltages applied to the diodes on the RRS elements, the refractive phase shifts can be reconfigured, so as to realize the desired beamforming. Different from the phased arrays, the RRSs do not contain any phase shifters, and the power consumption of each RRS element are ultra-low~\cite{M_2020}. Thus, the RRSs are more energy efficient than the phased array.

Existing works on metasurface-based wireless communications mainly focus on utilizing the metasurface as a passive relay. For example, in~\cite{SHBZL_2021}, a metasurface is deployed between a base station~(BS) and a user to extend cell coverage, and the location and orientation of the metasurface are jointly designed to maximize the cell coverage. In~\cite{BHLYZH_2020}, the authors investigate a multi-user MIMO system assisted by a metasurface, where the digital beamformer at the BS and the metasurface configuration are jointly optimized to maximize the sum rate. However, different from the existing works, the RRS acts as the transmit antenna in this letter. 

The use of the RRS as the transmit antenna brings new challenges. Specifically, unlike the metasurface working as the passive relay, the feed of the transmitter cannot be assumed to locate in the far field of the RRS since the RRS is much closer to the feed. Therefore, the distances from different RRS elements to the feed cannot be assumed to be equal, and thus the analysis of the data rate and power consumption of the RRS-aided system is challenging. To cope with the above challenges, in this letter, we consider a fundamental downlink network with one BS and one user equipment~(UE), where an RRS illuminated by a feed is deployed as the BS antenna. We first derive and analyze the data rate of the system. Then, the power consumption of the RRS is compared against the phased array under the same data rate requirement. Finally, simulation results validate our analysis and show that the RRS is an energy-efficient implementation of HMIMO.

\vspace{-0.2cm}
\section{System Model}
\subsection{Scenario Description}
\begin{figure}[!t]
	\centering
	\includegraphics[width=0.25\textwidth]{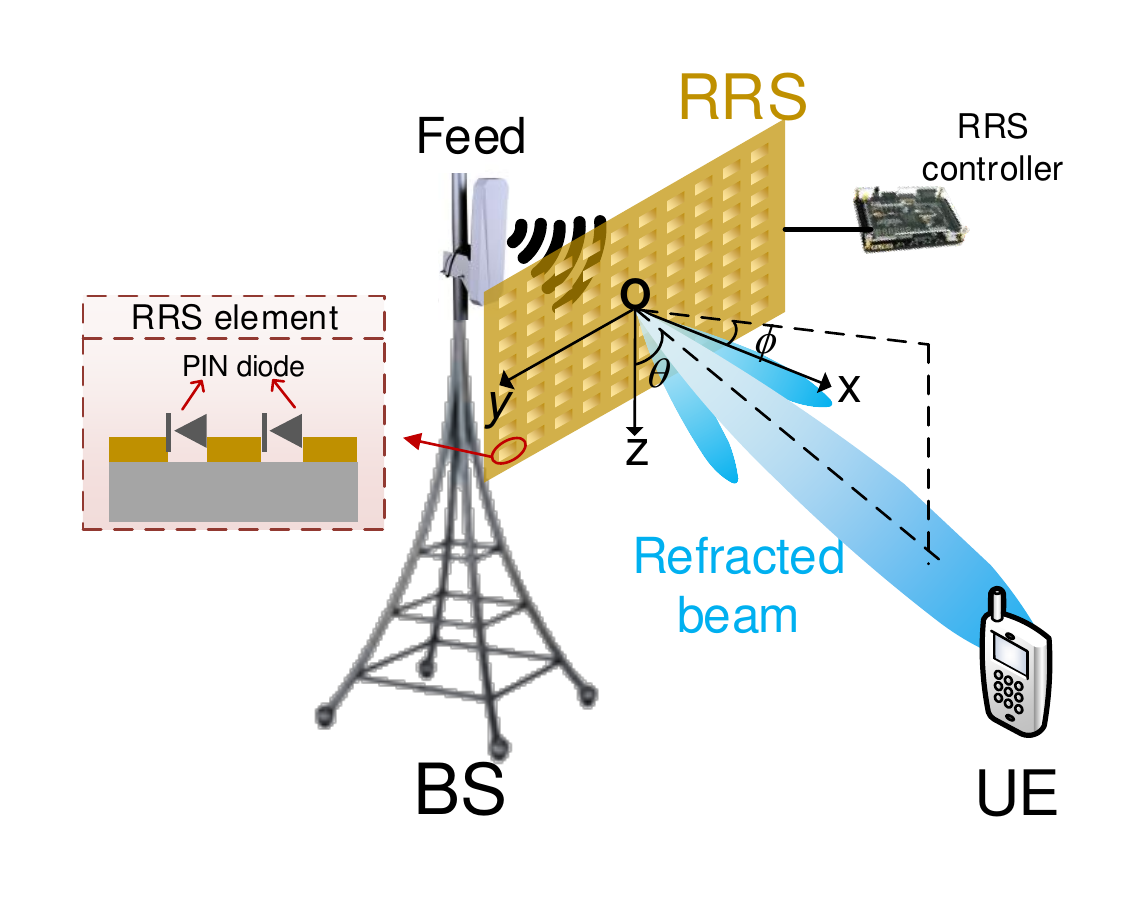}
	\vspace{-6mm}
	\caption{System model of a downlink single user network, where an RRS is used as the BS antenna.}
	\vspace{-5mm}
	\label{sysmodel}
\end{figure}
As shown in Fig.~\ref{sysmodel}, a narrow-band downlink network with one BS and one UE is considered. An RRS illuminated by a feed is utilized as the BS antenna to perform beamforming while an omin-directional antenna is adopted at the UE. For the sake of compactness, the feed is deployed in the near field of the surface. Besides, to point the transmitted wave of the feed towards the RRS, a directional antenna is used as the feed, whose radiation pattern can be given by~\cite{WG}
\begin{align}
\label{pattern_feed}
	G_F(\theta,\phi)=\left\{
	\begin{aligned}
		&2(\alpha+1)\cos^{\alpha}\theta,~\theta\in [0,\frac{\pi}{2}],\\
		&0,~\text{otherwise},
	\end{aligned}
	\right.
\end{align}
where $2(\alpha+1)$ denotes the gain of the feed. For the simplicity of discussions, we assume that the direction of the main beam is vertical to the RRS.

{To describe the topology of the network, we introduce a Cartesian coordinate as shown in Fig.~\ref{sysmodel}, where the $yoz$ plane coincides with the RRS, and the $x$-axis is vertical to the surface.} Define $r_F$ as the distance from the feed to the center of the RRS. By assuming that the feed is on the x-axis, the coordinate of the feed is given by $\bm{q}_F=[-r_F,0,0]^{\mathrm{T}}$. For the UE, let $r_U$ represent its distance to the center of the RRS, and $\theta_U$ and $\phi_U$ denote the zenith and azimuth angles, respectively.


\vspace{-0.3cm}
\subsection{Reconfigurable Refractive Surfaces}
\vspace{-0.1cm}
A reconfigurable refractive surface~(RRS) contains $M_R\times N_R$ sub-wavelength elements, each with the size of $s_{M,R}\times s_{N,R}$. Different from the traditional reflective metasurface element such as that in~\cite{BHLYZH_2020}, the RRS element does not contain the metallic ground in order to enable refraction. Each RRS element has several diodes onboard~\cite{W_transmitarray_2021}, which can be positive-intrinsic-negative~(PIN) or varactor diodes. An \emph{RRS controller} is utilized to apply different biased voltages to the diodes (i.e., \textit{ON} and \textit{OFF} states), which can influence the refraction phase shift of the element. Denote the refraction amplitude and phase shift of the $(m,n)$-th RRS element by $A$ and $\varphi^{(m,n)}$, respectively. Therefore, the transmission coefficient of the $(m,n)$-th RRS element is written by $\Gamma^{(m,n)}=A\exp{(j\varphi^{(m,n)})}$. 

Benefited from the reconfigurability, the RRS can generate controllable refractive beams. Specifically, when the signal transmitted by the feed impinges upon the RRS elements, it will be refracted by the elements with specific phase shifts. By controlling the applied phase shifts through configuring the elements' biased voltages, the radiated beam of the RRS can be directed accurately towards the UE with a high gain.
\vspace{-0.4cm}
\subsection{Signal Model}

\subsubsection{Refracted signals by the RRS elements} Define $G^{(m,n)}_F$ as the feed gain towards the direction of the $(m,n)$-th RRS element. Besides, let $D^{(m,n)}$ represent the distance between the feed and the $(m,n)$-th RRS element, and denote the wavelength corresponding to the carrier frequency as $\lambda$. Since the feed is in the near field of the RRS, the received signal at the $(m,n)$-th element from the feed can be given by~\cite{HY},
{\setlength\abovedisplayskip{0cm}
	\setlength\belowdisplayskip{0cm}
\begin{align}
\label{channel_F2R}
y_R^{(m,n)}=\bigg(\sqrt{\frac{G^{(m,n)}_FA_F^{(m,n)}}{4\pi (D^{(m,n)})^2}}\exp\left(-j\frac{2\pi}{\lambda}D^{(m,n)}\right)\bigg)x,
\end{align}
where $x$ is the transmitted signal of the feed, and $A_F^{(m,n)}$ is the projected aperture of the $(m,n)$-th RRS element towards the direction of the feed. Here, projected aperture $A_F^{(m,n)}$ can be further expressed as $A_F^{(m,n)}=(-\bm{u}_x)^{\mathrm{T}}(\bm{q}_F-\bm{q}^{(m,n)})s_{M,R}s_{N,R}/D^{(m,n)}$,
where $\hat{\bm{u}_x}$ represents the unit vector in the $x$ direction, and $\bm{q}^{(m,n)}$ is the location of the $(m,n)$-th RRS element. Afterwards, the RRS elements will refract the received signals, where the refracted signal of the $(m,n)$-th RRS element can be given by
\begin{align}
\label{refract_signal}
\widetilde{y}^{(m,n)}_R=y_R^{(m,n)}\Gamma^{(m,n)}.
\end{align}}
\vspace{-0.2cm}
\subsubsection{Channels from the RRS elements to the UE} To simplify analysis, we assume that the channels only contain pathloss. Since the size of the RRS can be extremely large, the UE can locate in either the near field or the far field of the RRS. When the UE is in the near field of the RRS, new characteristics are introduced into the channels compared with the far field case~\cite{HY}. Specifically, the wavefront corresponding to the UE can be more accurately described by the spherical wavefront. Besides, the variation of the received signal strength from different RRS elements cannot be ignored, and the projected aperture of different RRS elements is also unequal. To show these characteristics, we model the channel by
{\setlength\abovedisplayskip{0cm}
	\setlength\belowdisplayskip{0cm}
\begin{align}
\label{channel_R2U}
h_R^{(m,n)}=\sqrt{\frac{G_UA_U^{(m,n)}}{4\pi (d^{(m,n)})^2}}\exp\left(-j\frac{2\pi}{\lambda}d^{(m,n)}\right),
\end{align}
where $G_U$ is the antenna gain of the UE}, $d^{(m,n)}$ represents the distance between the UE and the $(m,n)$-th RRS element, and $A_U^{(m,n)}$ denotes the projected aperture of the $(m,n)$-th RRS element in the direction of the UE. By defining $\bm{q}_U$ as the location of the UE,  $A_U^{(m,n)}$ can be expressed as $A_U^{(m,n)}=(\bm{u}_x)^{\mathrm{T}}(\bm{q}_U-\bm{q}^{(m,n)})s_{M,R}s_{N,R}/d^{(m,n)}$.

When the refracted signals transmit through the channels from the RRS elements to the UE, the received signals at the UE can be written by 
{\setlength\abovedisplayskip{0cm}
	\setlength\belowdisplayskip{0cm}
\begin{align}
y=\sum_{m,n}h_R^{(m,n)}\widetilde{y}^{(m,n)}_R+n,
\end{align}
where $n$ represents the additive white Gaussian noise~(AWGN) at the UE with zero mean and $\sigma^2$ as variance. Therefore, the received signal-to-noise ratio~(SNR) at the UE is given by
\begin{align}
\label{SNR}
\gamma_R=|\sum_{m,n}h_R^{(m,n)}\widetilde{y}^{(m,n)}_R|^2\Big/\sigma^2.
\end{align}
\vspace{-0.3cm}
\section{Data Rate Analysis}
\vspace{-0.1cm}
\subsection{Derivation of Data Rate}
\vspace{-0.1cm}
Based on the SNR in (\ref{SNR}), the system data rate can be expressed as $C_R=\log_2(1+\frac{|\sum_{m,n}h_R^{(m,n)}\widetilde{y}^{(m,n)}_R|^2}{\sigma^2})$. Define $\Psi_U=\sin\theta_U\cos \phi_U$, $\Phi_U=\sin\theta_U\sin \phi_U$, and $\Omega_U=\cos\theta_U$. Besides, denote the transmit power by $P$. Since the data rate is correlated with the phase shifts of the RRS, it can be maximized by optimizing the phase shifts.
\vspace{-0.1cm}
\begin{theorem}
	\label{theorem_rate_RRS}
	The maximized data rate for the RRS is
	\begin{align}
	\label{max_rate}
	C_R=\log_2(1+L_R\left(\iint_{\mathcal{S}_R}{f_R(y,z)}dydz\right)^2),
	\end{align}
	where $L_R\!=\!\frac{PA^2G_UG_0\Psi_Ur_U^2}{\sigma^2(4\pi)^2r_F^2}$, $\mathcal{S}_R\!=\![-\frac{M_Rs_{M,R}}{2r_U},\!\frac{M_Rs_{M,R}}{2r_U}]\!\times\![-\frac{N_Rs_{N,R}}{2r_U},\!\frac{N_Rs_{N,R}}{2r_U}]$, and $f_R(y,z)\!=\!\Big(\!1\!+\!\frac{r_U^2}{r_F^2}(y^2\!+\!z^2)\Big)^{\!-\frac{\alpha+3}{4}}\!\Big(\!1\!-\!2\Phi_Uy\!-\!2\Omega_Uz\!+\!y^2\!+\!z^2\!\Big)^{\!-\frac{3}{4}}$

\end{theorem}
\begin{proof}
When the phases of different RRS-based channels are aligned, the data rate can be maximized as
{\setlength\abovedisplayskip{0cm}
	\setlength\belowdisplayskip{0cm}
	\begin{align}
	\label{RRS_maximize}
	C_R=\log_2\Big(1+\Big(\sum_{m,n}|h_R^{(m,n)}\widetilde{y}^{(m,n)}_R|\Big)^2\big/\sigma^2\Big),
	\end{align}
	which} can be further written as $C_R=\log_2(1+L_0(\sum_{m,n}f_R(\frac{ms_{M,R}}{r_U},\frac{ns_{N,R}}{r_U}))^2)$. Afterwards, following the method presented in Appendix A of~\cite{HY}, the maximum data rate can be transformed into the form in (\ref{max_rate}). 
\end{proof}
\vspace{-0.2cm}
{
\begin{remark}
	Theorem~\ref{theorem_rate_RRS} indicates that the RRS can influence the data rate via the number $M_RN_R$ of the RRS elements, the refraction amplitude $A$ of the elements, the gain $\alpha$ of the feed, and the distance $r_F$ between the feed and the RRS.
\end{remark}}
\vspace{-0.2cm}
Consider a circle $\mathcal{O}_R$ with the origin as center and $\min{(\frac{M_Rs_{M,R}}{2r_U},\frac{N_Rs_{N,R}}{2r_U})}$ as radius, which is found to be a part of region $\mathcal{S}_R$. Since $f_R(y,z)\ge 0$, a lower bound for $C_R$ can be derived by replacing the integral region $\mathcal{S}_R$ with $\mathcal{O}_R$, as shown in the following remark, which will be used for analyzing the power consumption of the RRS.
\begin{remark}
	Data rate $C_R$ can be lower bounded by
	{\setlength\abovedisplayskip{0cm}
		\setlength\belowdisplayskip{0cm}
	\begin{align}
	\label{rate_lower_bound}
	C_R\ge \log_2\Big(1+L_R\Big(\iint_{\mathcal{O}_R}{f_R(y,z)}dydz\Big)^2\Big)\triangleq C_R^{lb}.
	\end{align}}
\end{remark}
{
\begin{remark}
	\label{remark_lower_bound_achievable}
	The lower bound $C_R^{lb}$ in (\ref{rate_lower_bound}) is achievable when the number $M_RN_R$ of the RRS elements is sufficiently large while the gap to $C_R^{lb}$ is small when $M_RN_R$ is relatively small.
\end{remark}}
\vspace{-0.2cm}
For comparison, the data rate for the phased array is also derived. Assume that each antenna element within the phased array is omni-directional, whose antenna gain is denoted by $G_E$. Besides, use $M_PN_P$ and $(s_{M,P},s_{N,P})$ to represent the number of the phased array elements and the separations among the elements, respectively. 

\vspace{-0.2cm}
\begin{theorem}
	\label{theorem_rate_PA}
	The maximized data rate for the phased array under optimized phase shifters can be given by,
	{\setlength\abovedisplayskip{0cm}
		\setlength\belowdisplayskip{0.1cm}
	\begin{align}
	\label{rate_PA}
	C_P=\log_2(1+L_P(\iint_{\mathcal{S}_P}f_P(y,z)dydz)^2)
	\end{align}
	where} $L_P=\frac{P\lambda^2G_EG_Ur_U^2}{(4\pi)^2\sigma^2s_{M,P}^2s_{N,P}^2}$, $\mathcal{S}_P\!=\![\!-\!\frac{M_Ps_{M,P}}{2r_U}\!,\!\frac{M_Ps_{M,P}}{2r_U}\!]\!\times\![\!-\!\frac{N_Ps_{N,P}}{2r_U}\!,\!\frac{N_Ps_{N,P}}{2r_U}\!]$, and $f_P(y,z)=\frac{1}{\sqrt{M_PN_P}}\frac{1}{\sqrt{1+y^2+z^2-2\Phi_Uy-2\Omega_Uz}}$.
\end{theorem}
\begin{proof}
	See Appendix~\ref{app_PA_rate}.
\end{proof}
\vspace{-0.7cm}
\subsection{Discussions on Data Rate}
\vspace{-0.1cm}
{
\begin{lemma}
	\label{remark_limited}
	According to the derived data rate in (\ref{max_rate}), for an infinitely large RRS, the data rate does not increase unbounded, which is consistent with practical results. 
\end{lemma}}
\vspace{-0.4cm}
\begin{proof}
	See Appendix~\ref{app_scaling}.
\end{proof}
\vspace{-0.3cm}
Compared with the phased array, the RRS has several new degrees of freedom for design, i.e., the distance $r_F$ between the feed and the RRS, and the gain $\alpha$ of the feed. Therefore, the impacts of $r_F$ and $\alpha$ on the data rate are discussed.
\vspace{-0.2cm}
{
\begin{lemma}
	\label{rate_influence}
	When the number of the RRS elements is sufficiently large\footnote{A sufficiently large RRS is one that captures most power radiated by the feed.}, the data rate is positively correlated with the distance $r_F$ between the feed and the RRS, and is  negatively correlated with the gain $\alpha$ of the feed.
\end{lemma}}
\vspace{-0.4cm}
\begin{proof}
	Since the data rate $C_R$ monotonically increases with the received SNR $\gamma_R$ at the UE, we will show the influence of $r_F$ and $\alpha$ on $\gamma_R$ as follows. 
	
	Based on (\ref{RRS_maximize}), the received SNR can be rewritten as $\gamma_R=\big(\sum_{m,n}|h_R^{(m,n)}\Gamma^{(m,n)}||y_R^{(m,n)}|\big)^2\big/\sigma^2$. Then, by applying the Cauchy inequality, an upper bound for $\gamma_R$ can be acquired:
	\begin{align}
	\label{cauchy}
	\gamma_R\le \frac{1}{\sigma^2}(\sum_{m,n}|h_R^{(m,n)}\Gamma^{(m,n)}|^2)(\sum_{m,n}|y_R^{(m,n)}|^2).
	\end{align}
	Then, we show that when $r_F$ increases or $\alpha$ decreases, the gap between $\gamma_R$ and the derived upper bound in (\ref{cauchy}) becomes smaller. Specifically, the equality in (\ref{cauchy}) holds only when
\begin{align}
\frac{|y_R^{(m_1,n_1)}|}{|h_R^{(m_1,n_1)}\Gamma^{(m_1,n_1)}|}=\frac{|y_R^{(m_2,n_2)}|}{|h_R^{(m_2,n_2)}\Gamma^{(m_2,n_2)}|}.
\end{align}
	Denote the index of the center element of the RRS by $(0,0)$. Since the feed directs more power towards the center of the RRS than the edge of the RRS, and the feed is closer to the RRS as compared to the distance between the UE and the RRS, we have $\frac{|y_R^{(0,0)}|}{|h_R^{(0,0)}\Gamma^{(0,0)}|}>\frac{|y_R^{(m,n)}|}{|h_R^{(m,n)}\Gamma^{(m,n)}|}$ for $(m,n)\succ (0,0)$ based on (\ref{channel_F2R}) and (\ref{channel_R2U}). Therefore, there is gap between $\gamma_R$ and its upper bound in (\ref{cauchy}). However, according to  (\ref{pattern_feed}), when $\alpha$ becomes smaller, the radiation of the feed is more evenly distributed in different directions, and thus the gap can be narrowed. Besides, the gap can also be reduced by increasing the distance $r_F$ between the feed and the RRS. 
	
	Given the upper bound in (\ref{cauchy}) which is not related to $r_F$ and $\alpha$, we find that when $r_F$ increases or $\alpha$ decreases, the SNR $\gamma_R$ becomes larger, which verifies the remark. 
\end{proof}

%

\vspace{-0.5cm}
\section{Power Consumption Comparison}
In this part, we will compare the power consumption of the phased array and the RRS. For fairness, we assume that the required data rate for the RRS is equal to that for the phased array\footnote{For fairness, the dynamic transmit power for the RRS and that for the phased array are set equal, as well.}, which are represented by $C$. Thus, the data rate $C_R^{lb}$ for the RRS and $C_P$ for the phased array satisfy $C_R^{lb}=C_P=C$. To guarantee the quality-of-service~(QoS), we also set a minimum required data rate $C^{min}$, i.e., $C\ge C^{min}$. Moreover, it is assumed that the power consumption of the BS antenna cannot exceed $P^{max}$.
\vspace{-0.4cm}
\subsection{Power Consumption Model}
\vspace{-0.1cm}
In general, the power consumption of the BS antenna contains two parts, i.e., transmit power $P$ and the power consumed to support the operation of the antenna, denoted by $P_R$ for the RRS and $P_P$ for the phased array. {Since the dynamic transmit power for the RRS and that for the phased array are set equal, we only discuss $P_R$ and $P_P$ in the following.}
\subsubsection{Power consumption model for the RRS}
Based on the structure of the RRS-based antenna, we can find that power consumption $P_R$ mainly comes from the diodes on the RRS elements and the controller~\cite{W_transmitarray_2021}. Note that the controller is composed of an FPGA and several voltage converters~\cite{J_2019}, and both of them are active components. Therefore, we have
{\setlength\abovedisplayskip{0cm}
	\setlength\belowdisplayskip{0cm}
\begin{align}
\label{power_RRS}
P_R&=M_RN_RL_R^{(D)}P_R^{(D)}+\frac{M_RN_R}{Q}P_R^{(V)}+P_R^{(F)},
\end{align}
where $P_R^{(D)}$, $P_R^{(V)}$, and} $P_R^{(F)}$ are the power consumption of one diode, one voltage converter, and the FPGA, respectively, $L_R^{(D)}$ is the number of diodes within an RRS element, and $Q$ is the number of RRS elements within one group\footnote{For the ease of control, the metasurface is generally divided into several groups~\cite{HSBYMMLZH}. Each group contains a couple of elements sharing one voltage converter, and thus, the biased voltages for these elements are the same.}. 
\subsubsection{Power consumption model for the phased array}
For comparison, a model for the power consumption $P_P$ of the phased array is also required. The phased array generally contains one power divider, where each output of the power divider connects to one phase shifter followed by one antenna element. Besides, an FPGA is utilized to configure the phase shifters. Since the power divider and the antenna elements are passive, the total power consumption is
\begin{align}
\label{power_PA}
P_P=M_PN_PP_P^{(S)}+P_P^{(F)},
\end{align}
where $P_P^{(S)}$ and $P_P^{(F)}$ are the power consumption of one phase shifter and the FPGA, respectively. We assume that the power consumption of the FPGA in the RRS is equal to that in the phased array, i.e., $P_R^{(F)}=P_P^{(F)}$.
\vspace{-0.4cm}
\subsection{Comparison of Power Consumptions between the RRS and the Phased array}
\vspace{-0.1cm}
Our aim is to derive the range of the required data rate $C$ under which the RRS is more energy efficient. We would like to point out that  different $C$ can be achieved by changing the size of the RRS and the phase array. Define $l_{pw}$ as the ratio of power consumptions between one phased array element and one RRS element, i.e.,
{\setlength\abovedisplayskip{0cm}
	\setlength\belowdisplayskip{0cm}
 \begin{align}
 \label{ratio_ele}
 l_{pw}=P_P^{(S)}\Big/\Big(L_R^{(D)}P_R^{(D)}+P_R^{(V)}/Q\Big).
 \end{align}
According to (\ref{power_RRS}) and (\ref{power_PA})}, we have the following remark.
\vspace{-0.2cm}
\begin{remark}
	\label{the_power}
	When $\frac{M_RN_R}{M_PN_P}\!\le\! l_{pw}$, the RRS consumes less power than the phased array. Otherwise, the phased array has a smaller power consumption.
\end{remark}
\vspace{-0.2cm}
It can be found that the required data rate $C$ has an influence on $\frac{M_RN_R}{M_PN_P}$ while $l_{pw}$ is not correlated with $C$. Define $g$ as a function indicating the relation between $\frac{M_RN_R}{M_PN_P}$ and $C$, i.e., $g(C)=\frac{M_RN_R}{M_PN_P}$. The closed-form expression of $g(\cdot)$ is available only when the UE locates in the far field of BS antenna. This is because in the far field case, the distances from different antenna elements to the UE are approximately the same, and thus the expressions for the data rate in (\ref{rate_lower_bound}) and (\ref{rate_PA}) can be rewritten in a closed form, from which the closed-form expression of $g(\cdot)$ can be acquired. Therefore, the following discussions are separated into two parts, i.e., 1) the UE is in the far field of the RRS and the phased array, and 2) the UE locates in the near field of at least one antenna. 


\subsubsection{The UE locates in the far field}
The far field assumption indicates that the distance from the UE to the RRS (the phased array) is larger than the boundary of the far field of the RRS (the phased array), i.e., $r_U>2((M_Rs_{M,R})^2+(N_Rs_{N,R})^2)/\lambda$, and $r_U>2((M_Ps_{M,P})^2+(N_Ps_{N,P})^2)/\lambda$. Therefore, the number of the RRS elements and that of the phased array elements are limited by $M_RN_R < k_Rr_U\lambda/[2(k_R^2s_{M,R}^2+s_{N,R}^2)]\triangleq (M_RN_R)^{(thr)}$, and $M_PN_P< k_Pr_U\lambda/[2(k_P^2s_{M,P}^2+s_{N,P}^2)]\triangleq (M_PN_P)^{(thr)}$, respectively. Besides, according to (\ref{rate_lower_bound}) and (\ref{rate_PA}), the data rate $C_R^{(lb)}$ for the RRS and the data rate $C_P$ for the phased array are positively correlated with $M_RN_R$ and $M_PN_P$, respectively. Therefore, $C_R^{(lb)}$ and $C_P$ are smaller than their values under $(M_RN_R)^{(thr)}$ and $(M_PN_P)^{(thr)}$, respectively. Due to the equal data rate requirement between the RRS and the phased array, denoted by $C$, we have $C_R^{lb}=C_P\triangleq C$. Therefore, $C$ is upper bounded by
{\setlength\abovedisplayskip{0cm}
	\setlength\belowdisplayskip{0cm}
\begin{align}
C^{thr} \triangleq \min\{C_R^{lb}((M_RN_R)^{(thr)}),C_P((M_PN_P)^{(thr)})\}.
\end{align}
Recall that} a minimum required data rate $C^{min}$ is set in order to guarantee the QoS. As a result, we have $C\in [C^{min},C^{thr})$. It is worthwhile noting that when the UE is in the near field of the BS antenna, the required data rate $C$ is larger than $C^{thr}$, and thus $C^{thr}$ can be regarded as a threshold between the required data rate corresponding to the far field and that corresponding to the near field.

Without loss of generality, we assume that the dimension of the RRS in the $y$ direction is larger than that in the $z$ direction, i.e., $M_Rs_{M,R}> N_Rs_{N,R}$. Under the far field assumption, the data rate for the RRS in (\ref{rate_lower_bound}) and the data rate for the phased array in (\ref{rate_PA}) can be simplified into
{\setlength\abovedisplayskip{0cm}
	\setlength\belowdisplayskip{0cm}
\begin{align}
\label{lower_rate_RRS}
C_R^{lb}\!=\!\log_2\Big(1\!+\!\Big(\frac{4\pi L_R^{\frac{1}{2}}r_F^2}{(\alpha\!-\!1)r_U^2}\big(1\!-\!(1\!+\!\frac{M_RN_Rk_Rs_{N,R}^2}{4r_F^2})^{\frac{1\!-\!\alpha}{4}}\big)\Big)^2\Big),
\end{align}}
{\setlength\abovedisplayskip{-0.2cm}
	\setlength\belowdisplayskip{0cm}
\begin{align}
\label{rate_PA_re}
C_P=\log_2\big(1+(M_PN_Ps_{M,P}^2s_{N,P}^2L_P)/r_U^4\big),
\end{align}
respectively.} Recall that $C_R^{lb}=C_P\triangleq C$. Thus, from (\ref{lower_rate_RRS}) and (\ref{rate_PA_re}), we can derive the expression for the function $g(\cdot)$ indicating the relation between $\frac{M_RN_R}{M_PN_P}$ and $C$ as shown below,
{\setlength\abovedisplayskip{0cm}
	\setlength\belowdisplayskip{-0.2cm}
\begin{align}
\label{close_form}
g(C)\!=\!\frac{4s_{M,P}^2s_{N,P}^2L_Pr_F^2}{k_Rs_{N,R}^2r_U^4(2^{C}\!-\!1)}\Big[(1\!-\!\frac{(\alpha-1)r_U^2}{4\pi L_R^{0.5}r_F^2}\sqrt{2^{C}\!-\!1})^{-\frac{4}{\alpha\!-\!1}}\!-1\!\Big]\notag.
\end{align}}
\begin{remark}
	\label{mono}
	$g(C)$ first decreases and then increases with $C$. Besides, under sufficiently small or sufficiently large data rate, $g(C)$ is larger than $l_{pw}$ defined in (\ref{ratio_ele}).
\end{remark}
We can infer from Remark~\ref{mono} that there are only two data rate satisfying $g(C)=l_{pw}$, denoted by $C^{(e,1)}$ and $C^{(e,2)}$, respectively\footnote{In general, there exists data rate under which the RRS consumes less power than the phased array, i.e., $\exists C, s.t. g(C)<l_{pw}$.}. Without loss of generality, we  assume that $C^{(e,1)}<C^{(e,2)}$. Recall that the RRS consumes less power only when $\frac{M_RN_R}{M_PN_P}=g(C)<l_{pw}$ as indicated by Remark~\ref{the_power}, and that we only consider the required data rate within $[C^{min},C^{thr})$. Therefore, we have:
\vspace{-0.2cm}
\begin{theorem}
	\label{the_power_cmp}
	When $\max\{C^{min},C^{(e,1)}\}\le C\le \min\{C^{thr},C^{(e,2)}\}$, the RRS consumes less power. Otherwise, the phased array is more power-saving.
\end{theorem}
\vspace{-0.2cm}
\begin{remark}
	\label{the_power_small}
	Since the maximum of $g(C)=\frac{M_RN_R}{M_PN_P}$ within $[C^{min},C^{thr})$ is achieved at $C^{min}$ or $C^{thr}$, the RRS always consumes less power than the phased array when $l_{pw}>\max\{g(C^{min}),g(C^{thr})\}$ according to Remark~\ref{the_power}.
\end{remark}
\vspace{-0.2cm}


\subsubsection{The UE is in the near field}

Without the far field assumption, the closed-form expression of $g(C)$ is not available, and thus only qualitative analysis is presented in this part.
\vspace{-0.2cm}
\begin{theorem}
	\label{the_saved_power}
	Given sufficiently large required data rate $C$, by moving the feed away from the RRS or using a feed with a smaller gain $\alpha$, more power can be saved by the RRS compared with the phased array.
\end{theorem}
\vspace{-0.3cm}
\begin{proof}	
	Since the required data rate $C$ is given, by increasing the distance $r_F$ between the RRS and the feed or decreasing $\alpha$, {the size of the RRS under which the required data rate can be achieved becomes smaller according to Lemma~\ref{rate_influence}}. However, $r_F$ and $\alpha$ have no influence on the size of the phased array, which ends the proof.  
\end{proof}

\vspace{-0.5cm}
\section{Simulation Results}
\vspace{-0.1cm}
\begin{figure*}[!tpb]
	\centering
	\subfigure[]{
		\begin{minipage}[b]{0.25\textwidth}
			\centering
			\includegraphics[width=1\textwidth]{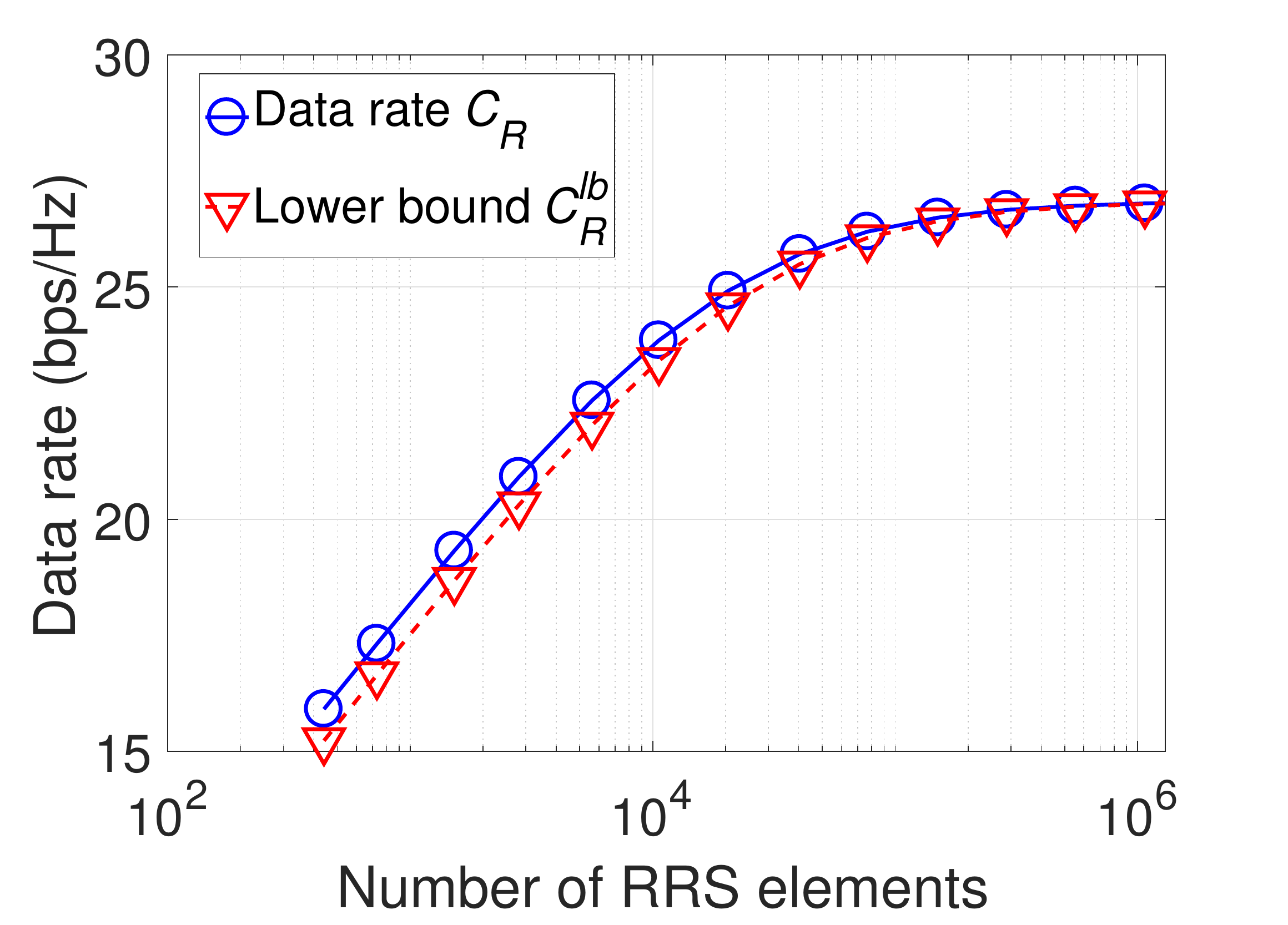}
			\vspace{-0.5cm}
			\label{lowerbound_vs_exact}
	\end{minipage}}
	\subfigure[]{
		\begin{minipage}[b]{0.245\textwidth}
			\includegraphics[width=1\textwidth]{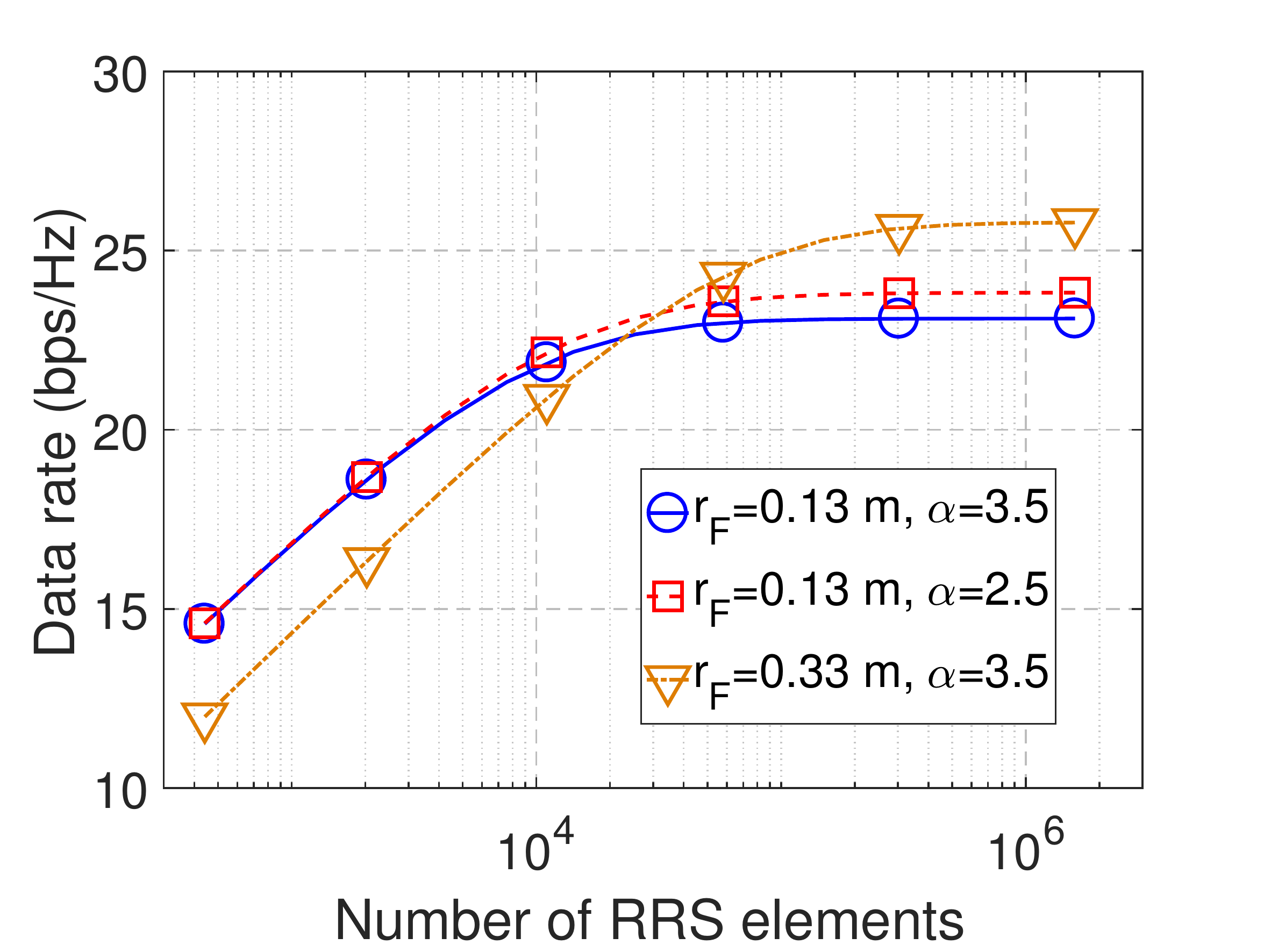}
			\vspace{-0.5cm}
			\label{a_capacity_vs_k}
	\end{minipage}}
	\subfigure[]{
		\begin{minipage}[b]{0.24\textwidth}
			\centering
			\includegraphics[width=1\textwidth]{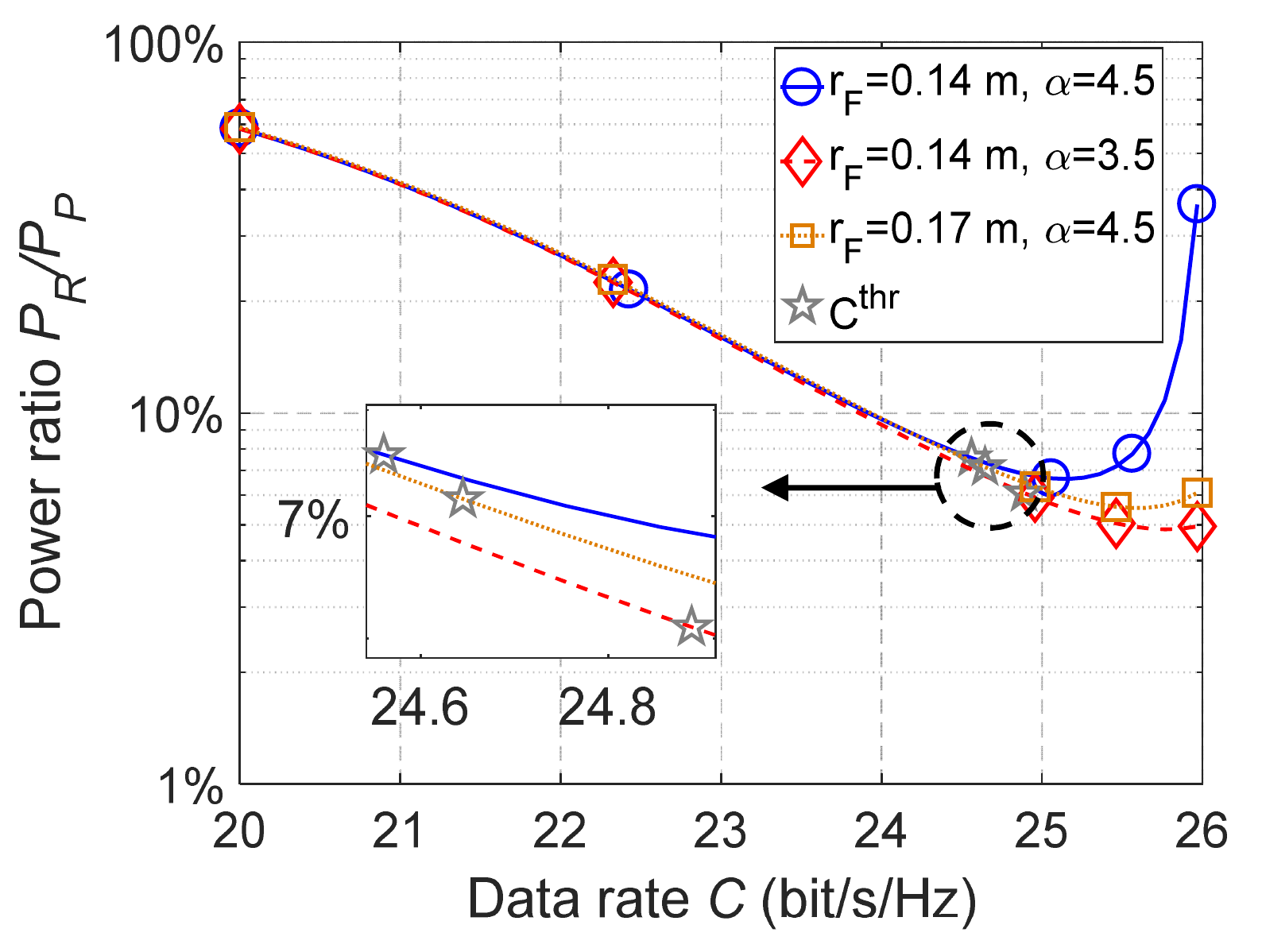}
			\vspace{-0.5cm}
			\label{energy_vs_C_UE_loc}
	\end{minipage}}
	\vspace{-0.3cm}
	\caption{{(a) Comparison between the exact data rate $C_R$ for the RRS and its lower bound $C_R^{lb}$ defined in (\ref{rate_lower_bound}), with with $r_F=0.15$~m and $\alpha=5$.} (b) {Data rate versus the number of the RRS elements.} (c) Power consumption ratio $\frac{P_R}{P_P}$ versus the required data rate $C$, where $C\le C^{thr}$ and $C>C^{thr}$ correspond to the far field case and the near field case, respectively.}
	\vspace{-0.6cm}
\end{figure*}
In this section, we verify the analytical results and evaluate how much power can be saved by the RRS through simulations. The parameters are based on existing works~\cite{SHBZL_2021}. Specifically, the center frequency is set as $26$~GHz, with a wavelength of $\lambda=1.15$~cm. The location of the UE is given by $(\theta_U,\phi_U,r_U)=(\frac{\pi}{6},\frac{\pi}{4},50~m)$. The transmit power and noise variance are $P=43$~dBm and $\sigma^2=-96$~dBm, respectively. We set $k_R=k_P=1$, i.e., $\frac{N_R}{M_R}=\frac{N_P}{M_P}=1$. The gain of the phased array antenna element is $G_E=1$ and that of the UE antenna is $G_U=1$. The size of the RRS elements and the element separation for the phased array are selected as $s_{M,R}=s_{N,R}=\frac{\lambda}{6}$, and $s_{M,P}=s_{N,P}=\frac{\lambda}{2}$, respectively. The refraction amplitude of the RRS elements is given by $A=0.8$. Each RRS element contains $L_R^{(D)}=1$ varactor diode of the power consumption $P_R^{(D)}=5\times 10^{-6}$~W. The power consumption of each phase shifter in the phased array, of one FPGA, and of one voltage converter are set as $P_P^{(S)}=0.1$~W, $P_R^{(F)}=P_P^{(F)}=5$~W, and $P_R^{(V)}=5\times 10^{-4}$~W, respectively. The maximum power constraint is $P^{max}=250$~W, and the minimum required data rate is $C^{min}=20$~bit/s/Hz.



Fig.~\ref{lowerbound_vs_exact} compares the data rate $C_R$ of the RRS against its lower bound $C_R^{lb}$. We can find that $C_R^{lb}$ is achievable when the number of the RRS elements is sufficiently large, while the gap between $C_R$ and $C_R^{lb}$ is small when the number of the RRS elements is relatively small, which verifies Remark~\ref{remark_lower_bound_achievable}. 

Fig.~\ref{a_capacity_vs_k} depicts the data rate versus the number $M_RN_R$ of the RRS elements. According to Fig.~\ref{a_capacity_vs_k}, the data rate will increase with $M_RN_R$ when $M_RN_R$ is relatively small. However, the data rate will converge when $M_RN_R$ continues to grow, {which verifies Lemma~\ref{remark_limited}}. From Fig.~\ref{a_capacity_vs_k}, we can also find that when $M_RN_R$ is sufficiently large, $r_F$ and $\alpha$ have a positive impact and a negative impact on the data rate, respectively, {which consists with Lemma~\ref{rate_influence}}. 



In Fig.~\ref{energy_vs_C_UE_loc}, we plot power consumption ratio $\frac{P_R}{P_P}$ versus the required data rate $C$. We can find that when $C<C^{thr}$, the RRS can always save energy compared with the phased array, which is consistent with Remark~\ref{the_power_small} given the fact that $l_{pw}$ is larger than $\max\{g(C^{min}),g(C^{thr})\}$. In addition, Fig.~\ref{energy_vs_C_UE_loc} shows that given sufficiently large data rate, by reducing the gain $\alpha$ of the feed or increasing the distance $r_F$ between the feed and the RRS, more power can be saved by the RRS compared with the phased array, which is consistent with Theorem~\ref{the_saved_power}. Fig.~\ref{energy_vs_C_UE_loc} also indicates that an RRS has the potential to significantly reduce power consumption compared with the phased array for any required data rate, and thus the RRS can serve as a practical enabler of HMIMO.


\vspace{-0.3cm}
\section{Conclusion}
\vspace{-0.1cm}
In this letter, we have considered a downlink network with one BS and one UE, where an RRS has been used as the BS antenna. The data rate of the RRS-aided system has been derived and analyzed. Then, we have compared the power consumption of the RRS against the phased array under the same data rate requirement. From the analytical and simulations results, we can conclude that: 1) When the RRS contains a sufficiently large number of elements, the data rate is positively correlated with the distance between the feed and the RRS, while it is negatively related with the gain of the feed. 2) We derive the range of the required data rate, by achieving which the RRS consumes less power than the phased array when the UE is in the far field of the BS antenna. 3) By designing the gain of the feed and its distance from the RRS, the RRS can significantly reduce the power consumption compared with the phased array for any required data rate, and thus the RRS is an energy-efficient way to holographic MIMO.


\begin{appendices}


\vspace{-0.5cm}
	\section{Proof of Theorem~\ref{theorem_rate_PA}}
	\vspace{-0.2cm}
	\label{app_PA_rate}
	For simplicity, we assume that the channels from the phased array elements to the UE only contains pathloss, which is described by the free space pathloss model. Therefore, the channel from the $(m,n)$-th element to the UE can be given by $h_P^{(m,n)}=\frac{\lambda\sqrt{G_EG_U}}{4\pi d^{(m,n)}}\exp(-j\frac{2\pi d^{(m,n)}}{\lambda})$~\cite{A_2005}.
	Besides, the received signal at the UE can be given by $y=\sum_{m,n}h_P^{(m,n)}\exp(j\varphi^{(m,n)}_P)\frac{1}{\sqrt{M_PN_P}}x+n$,
	where $x$ is the transmitted signal, $\varphi^{(m,n)}_P$ denotes the phase shift induced by the $(m,n)$-th phase shifter, and $n$ represents the received AWGN. Therefore, the data rate can be expressed as $C_P=\log_2(1+\frac{P\lambda^2G_EG_U}{M_PN_P(4\pi)^2\sigma^2}|\sum_{m,n}(d^{(m,n)})^{-1}\exp(j(-\frac{2\pi d^{(m,n)}}{\lambda}+\varphi^{(m,n)}_P))|^2)$.
	By aligning the phases of the received signals from different array elements, $C_P$ is maximized. Then, we can prove (\ref{rate_PA}) following the methods to prove Theorem~\ref{theorem_rate_RRS}.
\vspace{-0.5cm}
	\section{Proof of Lemma~\ref{remark_limited}}
	\vspace{-0.2cm}
	\label{app_scaling}
%
Note that $f_R(y,z)$ decreases with the antenna gain $\alpha$, and $f_R(y,z)> 0$. Thus, the data rate $C_R$ for the RRS is maximized when $\alpha=0$. Besides, based on the triangular inequality, we have $\big(\!1\!+\!\frac{r_U^2}{r_F^2}(y^2\!+\!z^2)\big)^{\frac{1}{2}}\ge |\sqrt{(\frac{r_U}{r_F})^2(y^2+z^2)}-1|$, and $\big(\!1\!-\!2\Phi_Uy\!-\!2\Omega_Uz\!+\!y^2\!+\!z^2\!\big)^{\frac{1}{2}}\ge |\sqrt{y^2+z^2}-1|$. Thus, we have 
{\setlength\abovedisplayskip{0cm}
	\setlength\belowdisplayskip{0cm}
\begin{align}
C_R\le \log_2(1+L_R\big(\iint_{\mathcal{S}_R}{f_R^{ub}(y,z)}dydz\big)^2)\triangleq C_R^{ub},
\end{align}
where $f_R^{ub}(y,z)=\!\big|\sqrt{(\frac{r_U}{r_F})^2(y^2+z^2)}-1\big|^{-3/2}\!\big|\sqrt{y^2+z^2}-1\big|^{-3/2}$. We can easily find out that $C_R^{ub}$ is limited when the number of the RRS elements is infinite.} Therefore, $C_R$ does not increase infinitely with the number of the RRS elements. 

\end{appendices}

\vspace{-0.4cm} 

\end{document}